\documentclass[aps, prd, twocolumn, showpacs, superscriptaddress, groupedaddress]{revtex4} 
\usepackage{newtxtext}
\usepackage{braket} 
\usepackage{stackrel} 
\usepackage{graphicx}
\usepackage{float}
\usepackage{mdframed}
\usepackage{amssymb}
\usepackage{amsmath}
\usepackage{mathtools} 
\usepackage{dcolumn}
\usepackage{ulem}
\usepackage{color}
\usepackage{bigints} 
\usepackage{bbm} 
\usepackage{amsthm} 
\usepackage{pifont}

\newtheorem{theorem}{Theorem} 
\newtheorem{conj}{Conjecture}
\newtheorem{prop}{Proposition}

\newtheorem{lemma}{Lemma}
\usepackage{subfigure, rotating, bm, array}
\theoremstyle{definition}

\usepackage[usernames,dvipsnames,svgnames,table]{xcolor}
\usepackage[pagebackref=false, colorlinks=true]{hyperref}
\hypersetup{
linkcolor=blue,     
citecolor=blue,     
urlcolor=blue}      

\begin{document}
\title{Spacetime singularities and caustic of outgoing null hypersurfaces}

\author{Ashok B. Joshi}
\email{gen.rel.joshi@gmail.com}
\affiliation{PDPIAS, Charotar University of Science and Technology, Anand- 388421 (Guj), India.}

\date{\today}

\begin{abstract}
In this paper, we investigate the caustic point developed by outgoing null hypersurfaces and the behavior of null curves. We analyze three cases: Future, past, and zero caustic point terminology for outgoing null hypersurfaces in the marginally bound Lemaitre-Tolman-Bondi (LTB) metric for homogenous and inhomogeneous cases defined in the present work. We show the criteria to form local and globally visible singularity in terms of the tangent on the apparent horizon curve.
\bigskip

$\boldsymbol{key words}$: Singularity, Causal structure of spacetime, Gravitational collapse
\end{abstract}
\maketitle

\section{Introduction}
Recently, there has been a lot of discussion on spacetime singularities. To explain spacetime singularity, a series of singularity theorems are given by Penrose and Hawking, which are most known as the Penrose-Hawking singularity theorem \cite{Hawking, wald}. The Raychaudhuri equation plays a significant role in the proof of the singularity theorem. For null congruence the Raychaudhuri equation implies
    \begin{equation}
        \frac{d\theta}{d\lambda} = - \frac{1}{2}\theta^{2} - \sigma^{ab}\sigma_{ab} - R_{ab}k^{a}k^{b} \leq 0.
    \end{equation}
From the above equation, a minimum requirement of the gravitational focusing is that $d\theta/d\lambda \leq -\frac{1}{2}\theta^{2}$ which means that the \textit{strong energy condition} hold, $R_{ab}k^{a}k^{b} \geq 0$, and null geodesics are hypersurface orthogonal. Which gives:
\begin{equation}
    \theta^{-1}(\lambda) \leq \theta^{-1}(0) +\frac{\lambda}{2}
\end{equation}
where $\theta(0)$ is the initial value of the expansion scalar at the affine parameter value, $\lambda = 0$. If initially $\theta(0) < 0$, then at an affine parameter $\lambda \to 2/ \lvert \theta(0) \rvert $, the expansion scalar of null congruence becomes a negative infinite, $\theta(\lambda) = - \infty$. At that point, some of the congruences come together which results in the formation of a caustic point.

The focusing theorem for the null geodesic congruence provides information about a caustic of the null congruence. However, in reference to the coordinate system, there are two types of null geodesics classified: outgoing and incoming. In a time-evolving collapsing spacetime, the strong energy condition holds which directly implies that the ingoing null geodesic always be ingoing throughout its journey. While there is a possibility that outgoing null geodesics turn back to the coordinate center. If the center is regular or we can say that free from spacetime singularity then the ingoing null congruences will become outgoing at a particular point is called a caustic point of null generators. The \textit{null generators} of the hypersurfaces are the null geodesics that lie within the hypersurfaces. The best example is caustic development in Minkowski spacetime \cite{Poisson}. Similarly, ingoing null geodesics that enter into the event horizon also develop caustic of the congruence of null generators which is also quite an interesting example of a caustic point \cite{Poisson}. A congruence of non-intersecting null geodesics runs within this hypersurface; these are known as the event horizon null generators. In other words, we can say that the caustic of null congruences does not imply that it develops spacetime singularity.

The Raychaudhuri equation is used to formulate the Cosmic Censorship Conjecture and singularity theorems. The singularity theorems do not necessarily imply that the singularities they anticipate are concealed behind an event horizon \cite{Hawking}. Instead, they provide information regarding singularity formation. The Cosmic Censorship Conjecture(CCC) has weak and strong versions of the theory \cite{Penrose:1969pc}. The weak theory of cosmic censorship hypotheses that there can never be a singularity that can be observed from future null infinity, i.e., that the singularity is never visible to distant observers in the universe. According to the strong version, singularities cannot also be locally naked. Local visibility is the ability of a family of null geodesics to leave the singularity but not to cross the boundary of the matter cloud, instead returning to the singularity.

The primary issue with CCC is that there are numerous exact solutions to the Einstein equations that have singularities that can be seen by observers. One of them is Reissner Nordstrom spacetime, the two-parameter family of static, asymptotically flat, spherically symmetric solutions to the Einstein-Maxwell equations. In the exact solution of the Einstein-Maxwell equation, the singularity is locally visible inside the Cauchy horizon. This example provides compelling evidence against the weak cosmic censorship hypothesis. However, Simpson and Penrose\cite{Simpson:1973ua} proposed that the Reissner-Nordström singularities are not generic. They pointed out that the inner horizon of the maximally extended manifold possesses blue-shift instabilities. Understanding the stability of the Cauchy horizon in the interior of Reissner-Nordstrom spacetime black holes served as the impetus for the research of weak null singularities \cite{Dafermos:2003vim, Dafermos:2003wr, Dafermos:2003yw, Luk:2013cqa}. Null singularity is also discussed in \cite{Bambhaniya:2021jum}, the Penrose diagrams of JMN1 spacetime (Joshi-Malafarina-Narayan spacetimes are first introduced in \cite{Joshi:2011zm, Joshi:2013dva}) are discussed, for $M_{0} > 2/3$ the central singularity is null-like, and for $M_{0} < 2/3$, the singularity is timelike. As discussed in the paper the first type of JMN1 spacetime always forms a shadow when the central singularity is null.

In Ortiz and Sarbach's paper \cite{Ortiz:2017gzr}, they show that ``the scalar field cannot diverge at the Cauchy horizon, except possibly at the central singular point" in the LTB spacetime. They assume that $\Phi$ is governed by a well-posed initial value problem on $D$ and that $\Phi$ possesses a stress-energy tensor $T = T_{\mu \nu} dx^{\mu} \otimes dx^{\nu}$ which, as a consequence of the equations of motion, is divergence-free and satisfies the dominant energy condition. In their paper, they analyze the stability of the Cauchy horizon associated with a globally naked singularity, a shell-focussing singularity arising from the complete gravitational collapse of a spherical dust cloud. Even their result implies that free-falling observers co-moving with the dust particles measure finite energy of the field, even as they cross the Cauchy horizon at points lying arbitrarily close to the central singularity.

\begin{figure*}
\centering
\subfigure[]
{\includegraphics[width=45mm]{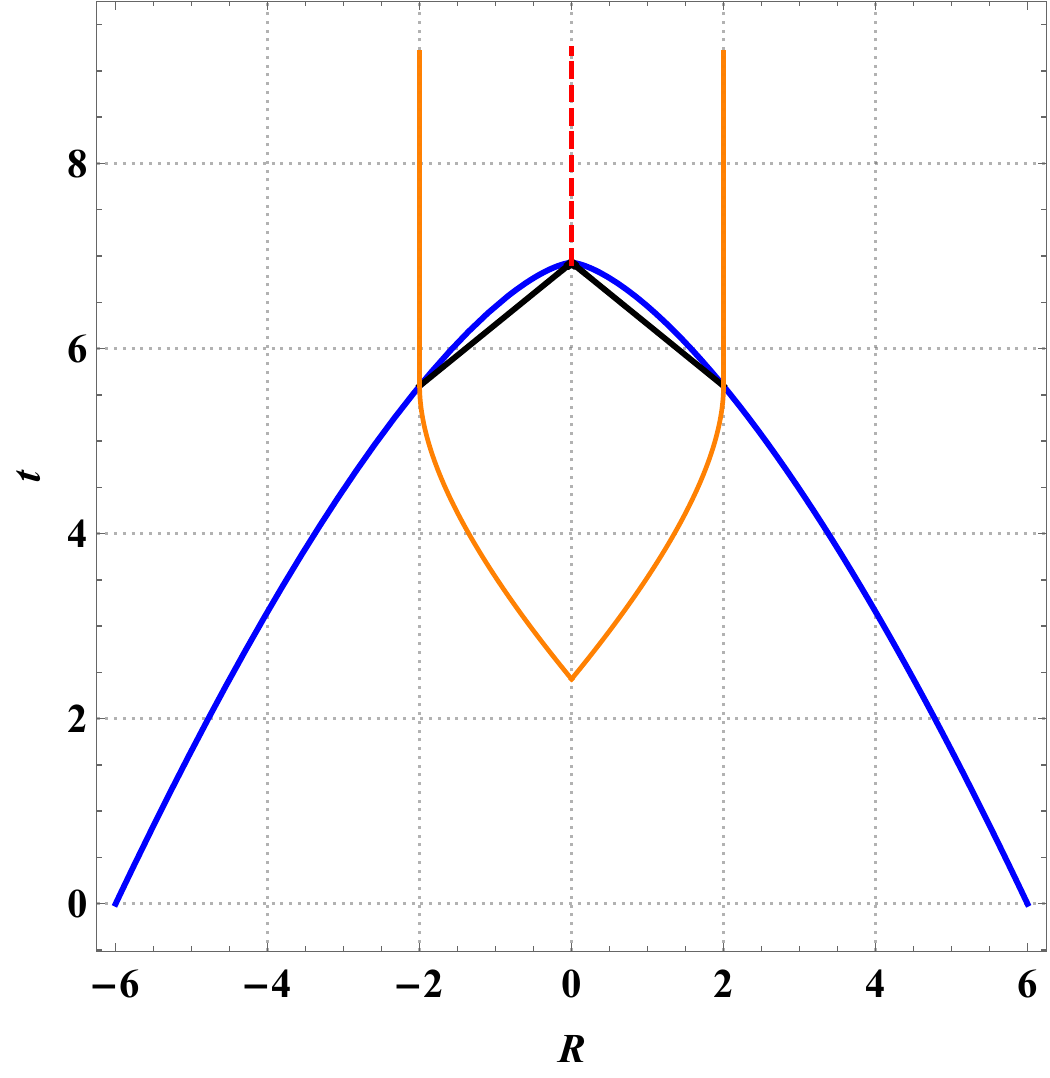}\label{ltbcollapse2}}
\hspace{0.5cm}
\subfigure[]
{\includegraphics[width=46mm]{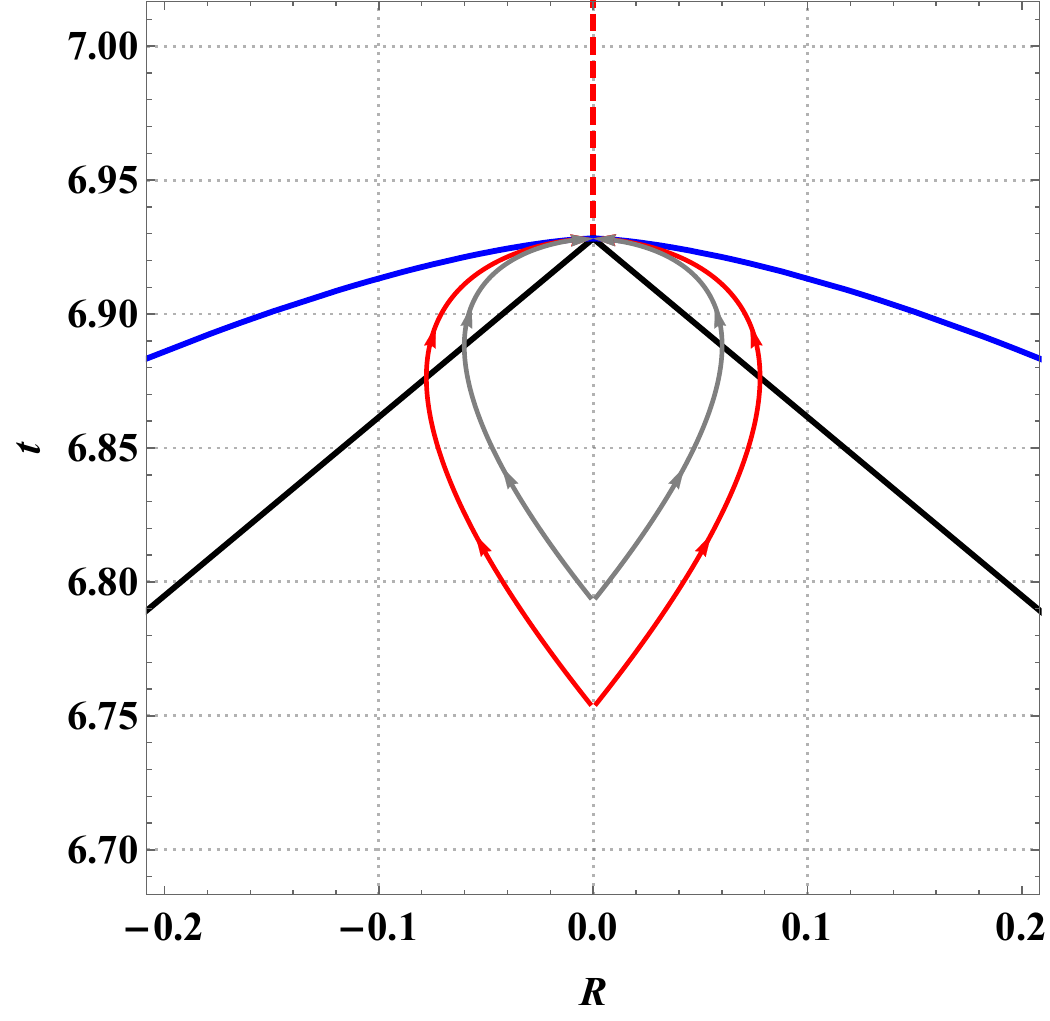}\label{ltbcollapsez1}}
\hspace{0.5cm}
\subfigure[]
{\includegraphics[width=65mm]{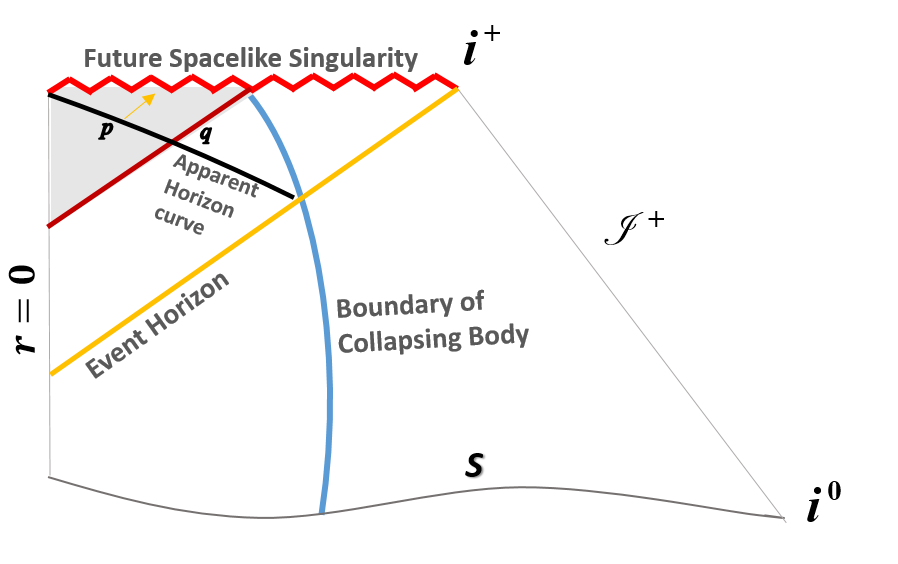}\label{penrose1}}
 \caption{This figure represents the gravitational collapse in the LTB metric, with the homogeneous density profile. The solid yellow line, blue line, and black lines represent the event horizon, the boundary of the collapsing cloud, and the apparent horizon respectively. While Fig.~(\ref{ltbcollapsez1}) is the zoomed-in picture of Fig.~(\ref{ltbcollapse2}), here red solid and gray color lines represent light rays emitted from $t_{p}$ time and with different value of $t_{o}$ time, respectively. The red dotted line is spacetime singularity. The time period, $t_{s}>t>t_{p}$, represents a caustic behavior in a homogeneous dust cloud.}\label{HomoLTB}
\end{figure*}

Here we consider marginally bound LTB metric as an interior collapsing body with Lorentzian manifolds, $\mathcal{M}_{int}$, while the outside matching metric is a Schwarzschild metric, $\mathcal{M}_{ext}$(exterior metric). The total smooth continuous Lorentzian manifold is $ \mathcal{M} = \mathcal{M}_{int} \cup \mathcal{M}_{ext}$. The LTB metric is the spherically symmetric dust solution of Einstein's field equations. Marginally bound LTB metric is a subclass of the LTB metric with the condition $f(r) = 0$ \cite{psjoshi2}. The line element of the marginally bound LTB metric is given by,
\begin{equation}
ds^2 = -dt^2 + R^{\prime 2}(t,r) dr^2 + R(r,t)^2 d\Omega ^2 \label{spacetime}
\end{equation}
where,$d\Omega^2 =\left(d\theta^{2} + sin^{2}\theta d\phi^{2}\right)$.  Energy momentum tensor of Eq.~(\ref{spacetime}), 
\begin{equation}
P_r = P_{\theta} = P_{\phi} = - \frac{\dot{F}(t,r)}{\dot{R}(t,r)R(t,r)^2} = 0,
\end{equation} 
this lead to $ F(r,t) \to F(r)$ and,
\begin{equation}
\rho(t,r) = \frac{F^{'}(r)}{ R^{'}(t,r) R(t,r)^2},\,\,\,\,\,\,\, F(r) = \dot{R}(t,r)^2 R(t,r). \label{ltbmassfun}
\end{equation}
From eq. (\ref{ltbmassfun}), the physical radius, $R(t,r)$, of the marginally bound collapsing object is given by,
\begin{equation}
R(r,t)  = \left( r^\frac{3}{2} - \frac{3}{2} \sqrt{F(r)}t\right)^\frac{2}{3}, \label{physicalr}
\end{equation}
where, $F(r)$ is the Misner-Sharp mass function, which describes the mass contained in the co-moving radius $r$. Spacetime singularity is formed when the curvature is become infinite which implies that the Kretsmann scalar becomes infinite. However, the Kretschmann scalar does not tell us anything about the nature of singularity. The Kretschmann scalar $\mathcal{K} = R^{ijkl}R_{ijkl}$, which is given as:
\begin{equation}
    \mathcal{K} = \frac{12 F^{'2}}{R^{4}R^{'2}} - \frac{32 F F^{'}}{R^{5}R^{'}} + \frac{48 F^{'2}}{R^{6}}.
\end{equation}
Here, $F = F(r)$ and $R = R(t,r)$. To identify the nature of singularity we take an approach that is tangent to the apparent horizon curve at the singularity point, $t_{s}(0)$. The curve of the apparent horizon plays a critical role in the causal structure of a collapsing spacetime, which provides information about local and global visibility. The apparent horizon curve in the collapsing spacetime is calculated using the expansion scalar of the outgoing null geodesics,
\begin{equation}
     \theta_{l} = \frac{2}{R(t,r)}\left(1 -\sqrt{\frac{F(t,r)}{R(t,r)}}\right) \label{apperentH}.
\end{equation}
The inhomogeneity of the matter distribution influences the curve of the apparent horizon\cite{Jhingan:1997ia}. As a result, inhomogeneity in dust collapse is important in the global and local visibility of singularity. It was previously demonstrated that a sufficiently inhomogeneous spherically symmetric dust cloud collapse results in a singularity, from which null geodesics can escape without being caught by the trapped surfaces\cite{Mosani:2020mro}. However, it was demonstrated by Eardley and Smarr\cite{Eardley:1978tr}, Dwivedi and Joshi \cite{Dwivedi:1996wf} that adding inhomogeneity to the collapsing cloud's mass distribution might alter how the apparent horizon evolved, potentially allowing non-spacelike geodesics to escape without becoming confined. From Eq.~(\ref{physicalr}) and (\ref{apperentH}), the singularity formation time and the apparent horizon curve (when the expansion scalar of outgoing null geodesic ($\theta_{l}$) is zero,$F(r) = R(t,r)$), are given by,
\begin{equation}
t_{s}(r) =  \frac{2r^{\frac{3}{2}}}{3\sqrt{F(r)}} 
\hspace{1.4cm}
t_{ah}(r) = \left(t_{s}(r) - \frac{2}{3} F(r)\right). \label{ltbtimeah}
\end{equation}
A collection of the \textit{marginally outer trapped surface} develops a curve in the LTB metric which is known to be an apparent horizon curve ($t_{ah}(r)$) given by Eq.~(\ref{ltbtimeah}).

The process of joining spacetimes via thin shell formalism across a boundary hypersurface is $\Sigma = \partial \mathcal{M}_{int} \cap \partial \mathcal{M}_{ext} $ \cite{Israel:1966rt, Huber:2019cze}. Two junction conditions must be satisfied on the matching hypersurface for the smooth matching of two spacetimes at a timelike or spacelike hypersurface \cite{Poisson}. By smooth matching, we mean the following: The first and second fundamental forms must be equal on both sides of the matching hypersurface as discussed in \cite{Poisson}. First, the matching hypersurface $\Sigma$ induced metric on both sides must be the same, and second, the matching hypersurface's extrinsic curvature of the internal and external spacetimes must be the same. The extrinsic curvature can be written as, 
\begin{equation}
    K_{ab} = e^{\alpha}_{a}e^{\beta}_{b}\nabla_{\alpha}\eta_{\beta},
\end{equation}
where, $\eta^{\beta}$ and $e^{\alpha}_{a}$ are the normal and tangents to the hypersurfaces, respectively.
So we can write co-moving coordinates in the Eddington-Finkelstein metric,
\begin{equation}
    ds^2 = -\left(1 - \frac{2M_{ADM}}{\mathcal{R}(\nu)} + 2\frac{d\mathcal{R}(\nu)}{d\nu}\right)d\nu^2 + \mathcal{R}(\nu)^2 d\Omega^2 \label{schmatch}
\end{equation}
Here, $\nu$ is the retarded null coordinate, $M_{ADM}$ is the ADM mass of spacetime and $\mathcal{R}(\nu)$ is the evolving radius \cite{Goswami:2007na}. 
On matching Eq.~(\ref{schmatch}) with Eq.(\ref{spacetime}) at the hypersurface $r = R_{b}$, we get,
\begin{equation}
    \mathcal{R}(\nu) = R(t, R_{b}) = R_{b} \, a(t),
\end{equation}
\begin{equation}
    F(R_{b}) = 2M_{ADM}, \label{adm}
\end{equation}
\begin{equation}
    \ddot{\mathcal{R}}(\nu) = -\frac{F(R_{b})}{2\mathcal{R}(\nu)^2} = -\frac{M_{ADM}}{\mathcal{R}(\nu)^2}.
\end{equation}
For the homogeneous case, Eq.(\ref{adm}) gives, $M_{o} = (2M_{ADM})/R_{b}^{3}$. Here we have freedom in two parameters of space, one is the total mass of the collapsing body ($M_{ADM}$) and the other is the boundary of the collapsing body (outermost co-moving shell), $R_{b}$. For the inhomogeneous case, for the mass function $F(r) = M_{o}r^{3} - M_{3}r^{6}$ matching can provide the values of $M_{o}$ and $M_{3}$ as follows:
\begin{equation}
    M_{o} = \frac{4M_{ADM}}{R_{b}^3} - \frac{1}{3}\rho(0,R_{b})
\end{equation}
\begin{equation}
     M_{3} = \frac{2M_{ADM}}{R_{b}^6} - \frac{\rho(0,R_{b})}{3R_{b}^3}.
\end{equation}
Because the mass function is a positive and monotonically increasing function, it suggests that $M_{o}>0$ and $M_{3}>0$. This leads us to the conclusion that the initial density at the boundary radius must satisfy the requirement that $\rho(0, R_{b}) < 6M/R_{b}^{3}$.

Null geodesics are important in studying global and local visibility of the singularity. The radially ingoing and outgoing light-like geodesics for the Eq.~(\ref{spacetime}), $k^{\mu}k_{\mu} = 0$, we get,
\begin{equation}
    \frac{dt_{null}}{dr} = \pm R^{'}(t,r). \label{nullgeodesicsltb}
\end{equation}
Here plus, and minus correspond to outgoing and incoming null geodesics respectively. One of the simple way to solve the differential Eq.~(\ref{nullgeodesicsltb}) around the coordinate center is given as follow:
\begin{equation}
   t_{null}(r) = t_{0} + \chi_1(t_{0}) r + \chi_2(t_{0}) r^{2} + \chi_3(t_{0}) r^3 + \mathcal{O}(r^4)
\end{equation}
where,
\begin{equation}
    \chi_i(t_{0})=\frac{1}{i!}\frac{\partial ^i t(r)}{\partial r^i}\bigg |_{r=0}
\end{equation}
where $t_{0}$ is the initial comoving time for outgoing null geodesic which originates from the coordinate center, $r =0$. A null geodesic which stratifies $t(R_{b}) = t_{ah}(R_{b})$ is the \textit{first outermost marginally outer trapped surface} which is known as the event horizon.

In this paper we investigate the behavior of null curves and tangent on the apparent horizon curve to identify the nature of initially form spacetime singularity (which means $t_{s}(0)$). In the marginally bound LTB spacetime, the caustic formation in the outgoing null hypersurfaces (caustics form during the formation of a spacetime singularity) is analyzed. We classified three cases using the nature of the outgoing null hypersurfaces: past-caustic point, future-caustic point, and zero-caustic point. Using the caustic type we are classifying the nature of spacetime singularity. The investigation of the null geodesics in three cases of inhomogeneity reveals that the singularity's nature shifts from hidden to local and finally, global visibility as inhomogeneity increases, causing the transition from the `future-caustic point of outgoing null hypersurfaces' to the `zero-caustic point of outgoing null hypersurfaces' and finally, the `Past-caustic point of outgoing null hypersurfaces'. In this context, the term ``Futuer-caustic point of outgoing null hypersurfaces" is described as two or more null hypersurfaces that originated from constant time slices with two or more different physical radii and intersect at one common point in the future. While a ``Past-caustic point of outgoing null hypersurfaces" is described as emitting multiple outgoing null hypersurfaces coming out from a single spacetime point at the same time. We demonstrate that the development of a \textit{Future-caustic point of outgoing null hypersurfaces} is a direct indication of black hole development. A \textit{Past-caustic point of outgoing null congruences}, suggests both a local and a globally visible singularity.

The plan of the paper is as follows. In section (\ref{sec2}), we discuss the geometry of null hypersurfaces, and we derive the radial motion of null geodesics and tangent on the apparent horizon to classify hidden and local visible singularity, for the homogeneous and inhomogeneous cases. We discuss the example of three types of caustic points using the radial motion of outgoing hypersurfaces. In section (\ref{sec3}), we discuss the genericity of zero caustic point. Finally, in section (\ref{result}), we discuss our results. Throughout the paper, we take $G=c=1$.


\section{Spacetime singularity and Caustic in outgoing null hypersurfaces}\label{sec2}
A spacetime $(\mathcal{M}, g)$ is a marginally bound LTB spacetime, and metric on the $\mathcal{M}$ is $g$. At any point $p \in \mathcal{M}$, a null vector in the tangent space $V \in T_{p}\mathcal{M}$ satisfy the condition,
    \begin{equation}
       g(V, V) = g_{ab}V^{a}V^{b} = 0.
    \end{equation}
The null vectors at $p$ develop a light cone in the tangent space which separates timelike and spacelike vectors. Let's consider a null hypersurface $\Sigma$, and $p,q \in \Sigma$ such that $q \in J^{+}(p) \backslash I^{+}(p)$ which directly implies that $p$ and $q$ is connected by null geodesics $\gamma(\lambda)$ within the hypersurface. Their normal vector is tangent to the null hypersurface. The collection of all the null geodesics congruences that lie within the hypersurface are the null generators of the hypersurface. That implies, at each point $p \in \Sigma$ the induced metric is degenerate.


\begin{figure*}[]
\centering
\subfigure[]
{\includegraphics[width=50mm]{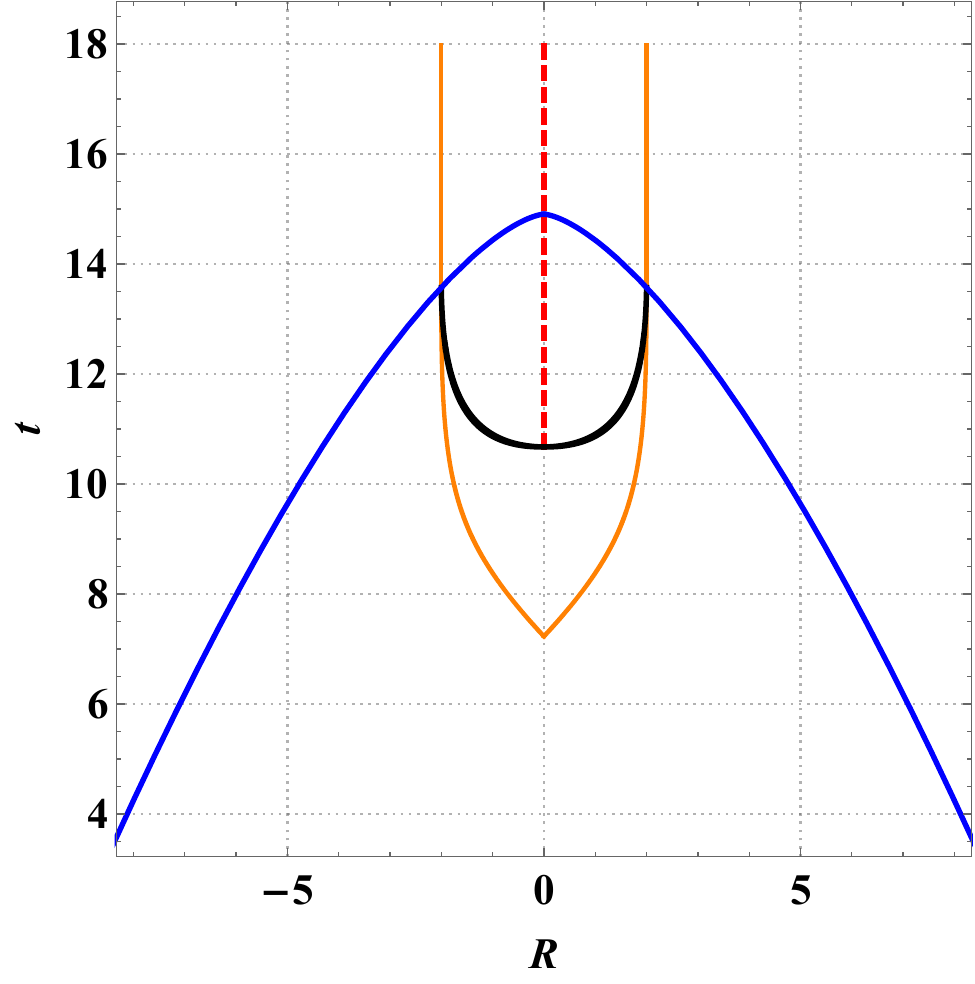}\label{ltbcollapseInHomo3}}
\hspace{0.5cm}
\subfigure[]
{\includegraphics[width=50mm]{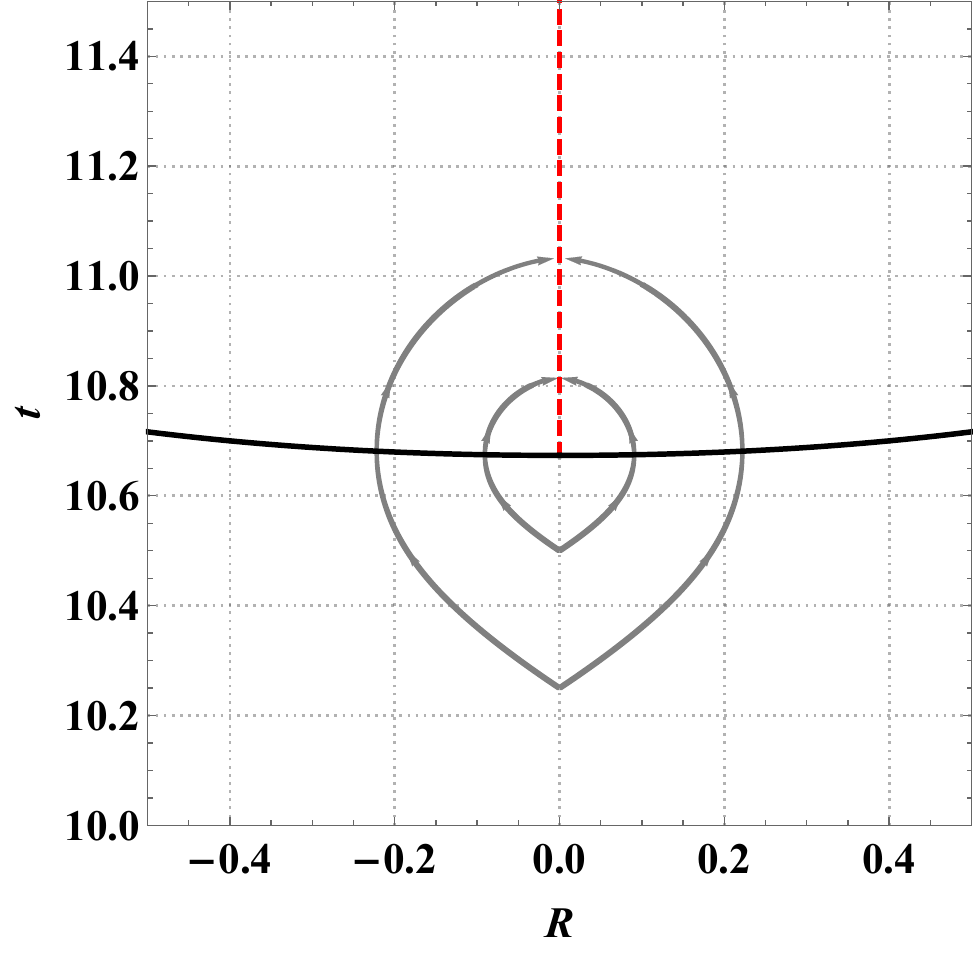}\label{ltbcollapseInHomo4}}
\hspace{0.5cm}
\subfigure[]
{\includegraphics[width=65mm]{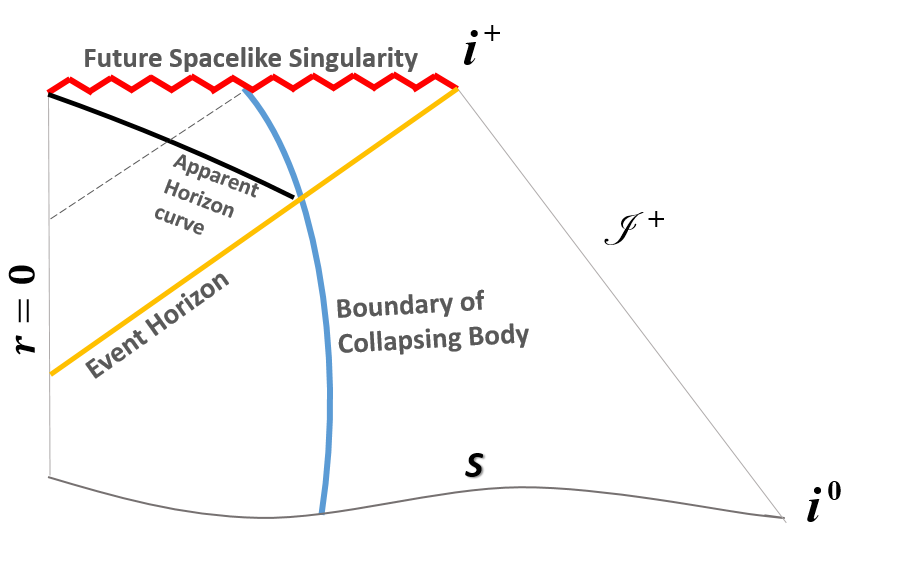}\label{penrose3}}\\
\subfigure[]
{\includegraphics[width=49mm]{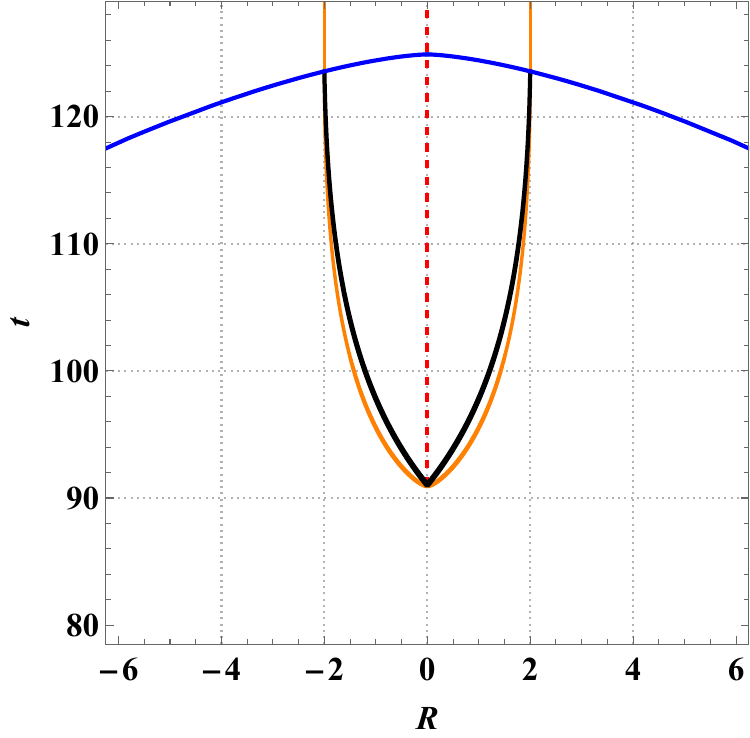}\label{ltbcollapseInHomo5}}
\hspace{0.5cm}
\subfigure[]
{\includegraphics[width=49mm]{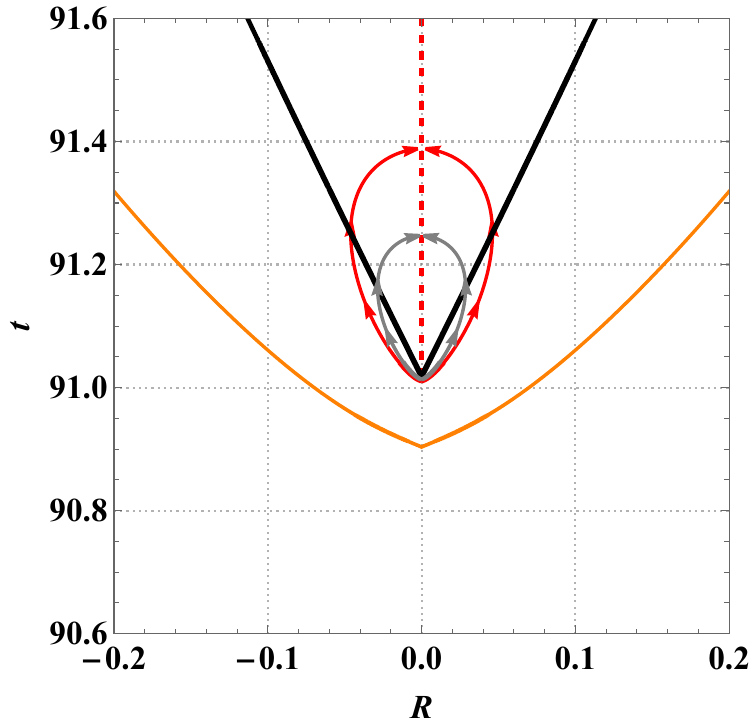}\label{ltbcollapseInHomo6}}
\hspace{0.5cm}
\subfigure[]
{\includegraphics[width=68mm]{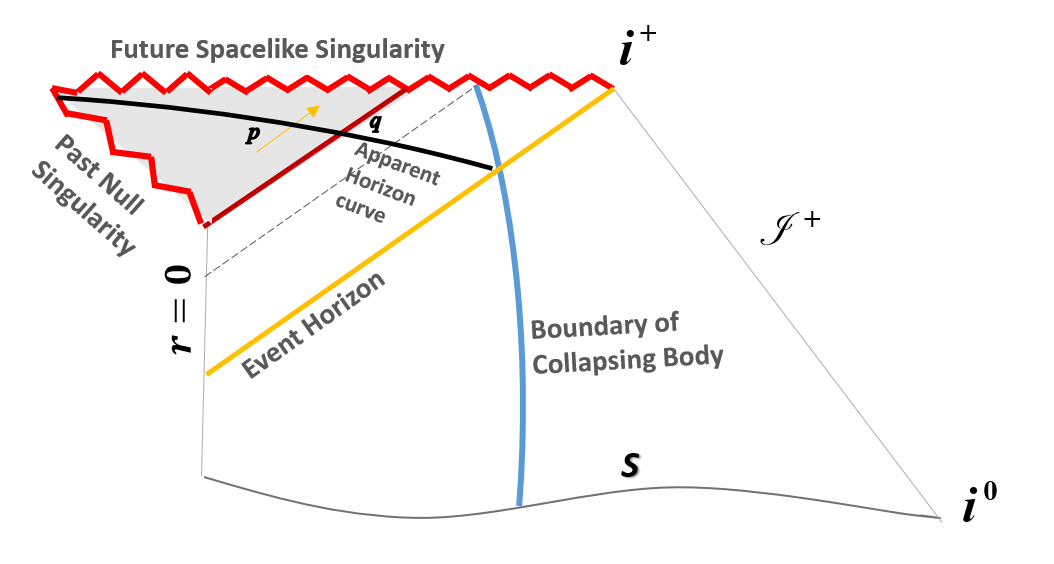}\label{penrose4}}
 \caption{This figure represents the gravitational collapse in the LTB metric, with the homogeneous density profile. The solid yellow line, blue line, red line, and black lines represent the event horizon, the boundary of the collapsing cloud, the Cauchy horizon, and the apparent horizon respectively. Past null singularity is the point in the spacetime diagram while in the conformal diagram,  it is represented as a red sig-sag line, Fig.~(\ref{penrose4}). In this figure, gravitational collapse in LTB metric is shown with the mass function $F(r) = M_{o}r^{3} - M_{3}r^{6}$. For $M =1$, $R_{b}=10M$,$\rho = 2.96* 10^{-4}$, the tangent on apparent horizon at $r=0$ have zero value as show in \ref{ltbcollapseInHomo3}. For locally naked, tangent has a value greater than 1 as shown in \ref{ltbcollapseInHomo5} with parameter, $M =1$, $R_{b}=41.25M$, $\rho = 10^{-5}$.}\label{InHomo}
\end{figure*}
We say that spacetime is to be singular if it contains an incomplete non-spacelike geodesic, $\gamma : [\ 0, t) \rightarrow \mathcal{M}$ such that there is no extension $\theta: \mathcal{M} \to \mathcal{M}^{'}$ for which $\theta \circ \gamma$ is extendible. In marginally bound LTB case in terms of geodesic incompleteness, two types of spacetime singularity are classified: \textit{future spacelike} and \textit{past null singularity} \cite{Joshi:2023ugm, Joshi:2020tlq}. Future spacelike singularity is hidden for local and observer at future null infinity. However, past null singularity is globally or locally visible. If null geodesics is future and past inextendable that is both ends have spacetime singularity which implies local visibility of singularity. In a closed gamma curve, if both ends are geodesically incomplete, the future end of the geodesic is future spacelike singularity while the past endpoint is past null singularity in marginally bound LTB spacetime. 
The maximal analytic extension of the metric provides information about the central singularity, where the curvature becomes infinite at $r=0$ which fails to represent spacetime geometry. Therefore, singularity $\mathcal{S}$, $\mathcal{S} \cap \mathcal{M} = \emptyset$, which is not a part of the manifold. Timelike singularity (Naked singularity) is connected to the timelike region,$J^{-}(\mathcal{J}^{+})$, $\mathcal{S} \subset J^{-}(\mathcal{J}^{+})$ which implies that the spacetime is not globally hyperbolic.
If the timelike region, $J^{-}(\mathcal{J}^{+})$, is not connected with singularity, $\mathcal{S}$, which means $\mathcal{S} \cap J^{-}(\mathcal{J}^{+}) = \emptyset$ which implies that $\exists$ event horizon $\mathcal{H}(= \partial J^{-}(\mathcal{J}^{+}))$ in the manifold, iff $\mathcal{S} \neq \partial J^{-}(\mathcal{J}^{+})$. The topology of the event horizon in Schwarzschild spacetime is $\mathbb{S}^{2} \times \mathbb{R}$ and is embedded null hypersurface in the manifold, $\mathcal{M}$.  In Schwarzschild spacetime, the black hole region, $\mathcal{B} (=  \mathcal{M} \setminus J^{-}(\mathcal{J}^{+}))$, of the spacetime manifold $\mathcal{M}$, and the event horizon $\mathcal{H}$ cover a singularity. 

The end state of inhomogeneous dust collapse predicts that singularity could give both possibilities that is, singularity is hidden or globally visible to the asymptotic observer. A singularity contained in Schwarzschild's black hole spacetime is called a future-spacelike singularity. A singularity contained in Schwarzschild's spacetime is a spacelike singularity. Spacelike singularity is subdivided into two parts: future and past spacelike singularity \cite{Joshi:2023ugm}

\begin{lemma}
    Let us consider an outgoing null geodesic $\gamma(t_{null})$, originate from center ($r=0$) with comoving time $t_{0}$, and at $\theta_{l} = 0$ comoving time is $t_{1}$, then $t_{1}$ is always greater then $t_{0}$ for all $\gamma$. Where, $t_{0},t_{1} \in t_{null}$.
\end{lemma}
\begin{proof}
    The expansion scalar of outgoing null congruences is being used to identify the causal structure of spacetime geometry. In the marginally bound LTB case, $\dot{R}(t,r)$ for the null geodesic can be written as follows:
\begin{equation}
    \dot{R}(t,r) = R(t,r) \theta_{l} = \left(1 -  \sqrt{\frac{F(r)}{R(t,r)}}\right).\label{rdot}
\end{equation}
Hence, from Eq.~(\ref{rdot}),
\begin{equation}
    \frac{dt_{null}}{dR} = \left(1 -  \sqrt{\frac{F(r)}{R(t,r)}}\right)^{-1} \label{dtnulldR}
\end{equation}
Where over-dot on $R$ represents derivative with respect to the comoving time, $t$. The congruence of light-like geodesics on the apparent horizon $\theta_{l}=0$, $\dot{R}(t(r),r)=0$, while for timelike $\dot{R}(t,r)= -1$ (as we can see from Eq.~(\ref{ltbmassfun})). While, at center $r = 0$ and $t = t_{0}$, $\Dot{R} = 1$ for $t_{o} < t_{s}$. At a coordinate center for all outgoing null geodesic $\theta_{l}$ is positive infinity that is a caustic of null generators. Therefore, at $\theta_{l}=0$, marginally outer trapped surfaces are the turning points at which radially outgoing light turns back to the coordinate center. While at $\dot{R}(t,r) < 0$, trapped surfaces are present. Therefore, if $\gamma$ is smooth continuous, and the tangent value is monotonically decreasing, then we can analyze that $t_{o} < t_{1}$. As we see the relationship between the expansion scalar and null geodesics from Eq.~(\ref{rdot}),  we analyze that $t_{o} < t_{1}$.
\end{proof}

\begin{lemma}
    If $\gamma$ is an outgoing null geodesic and the expansion scalar of the outgoing null geodesic at $R \to 0$ is $\theta_{l} \to -\infty$, then $\frac{dt_{null}}{dR} \to 0$.
\end{lemma}
\begin{proof}
    From the Eq.~(\ref{rdot}) and Eq.~(\ref{dtnulldR}), if $F(r) << R(t,r)$, at the future end of outgoing null geodesic, when $R(t,r) \to 0$ and $\theta_{l} \to -\infty$ gives $\dot{R}(t,r) \to - \infty$, that directly implies that $\frac{dt_{null}}{dR} \to 0$. In the same case, at $R(t,r) \to 0$, $\theta_{n} \to - \infty$, which suggests that the light cone close up at $R = 0$. That will give future spacelike singularity.
\end{proof}
\begin{prop}
    Consider a marginally bound LTB spacetime $(\mathcal{M}, g)$ seeded by the pressureless fluid, and if it contains a Cauchy horizon ($\mathcal{CH}^{+}$) within the spacetime manifold is outermost past inextendible null hypersurface. Then, $\forall \gamma$ within the $\mathcal{CH}^{+}$ is past inextendible.  Hence, if $p \in \mathcal{CH}^{+}$ and also $p \in \gamma$ then $\gamma$ is past inextendible. 
\end{prop}
\begin{proof}
In marginally bound LTB spacetime, the Cauchy horizon is a null hypersurface with a topology $\mathbb{S}^2 \times \mathbb{R}^2$. The caustic developed by null generators for the Cauchy horizon develops spacetime singularity. Let's consider a closed achronal set, S. The future Cauchy horizon of $S$ is defined as follows:
\begin{equation}
    \mathcal{CH}^{+}(S) = \overline{D^{+}(S)} \setminus I^{-}[D^{+}(S)],\label{CauchyD}
\end{equation}
and similar analogy for $\mathcal{CH}^{-}(S)$. Where, $D^{+}(S)$ is the \textit{future Cauchy development} of S. If a set S is achronal and $C^{0}$-submanifold of $\mathcal{M}$. There do not exist $p,\,q \in S$ such that $q \in I^{+}(p)$. The proof of this is similar to that of the theorem given in \cite{wald} (theorem 8.3.5). So, here we do not discuss the detailed proof but will discuss an interesting concept associated with a tangent on the apparent horizon curve, $\mathcal{A}$. Now, from the definition of Cauchy horizon, it's clear that if $p \in \mathcal{CH}^{+}(S)$ and $ q \in I^{-}(p)$ then $\gamma$ passing through $q$ is past and future complete.
\end{proof}

Mainly three cases are widely discussed in LTB spacetime: hidden, locally visible, and globally visible spacetime singularity. Locally visible and hidden singularities are covered within the event horizon. Here it is explained by two concepts, the tangent on null geodesics, and the tangent on the apparent horizon curve at the coordinate center.  The position and dynamics of the apparent horizon determine the global and local visibility of spacetime in the marginally bound LTB collapse. The global and local visibility information is thus provided by the tangent field on the apparent horizon at r = 0. From eq.(\ref{ltbtimeah}):
\begin{equation}
   \frac{d t_{ah}}{d R} = \lim_{r\to 0} \left(-\frac{2}{3} + \frac{\sqrt{r}(3F(r) - r F^{'}(r))}{3F(r)^{3/2} F^{'}(r)}\right)
\end{equation}
For example, let us consider an in-homogeneous mass profile $F(r) = M_{o}r^{3} - M_{3}r^{6}$, the tangent field on the apparent horizon can be written as follows:
\begin{equation}
   \frac{dt_{ah}}{dR}\Big|_{r = 0} = \frac{1}{3}\left(-2 + \frac{M_{3}}{M_{0}^{5/2}}\right).\label{tangentah}
\end{equation}

\begin{lemma}
    If $p \in \mathcal{A}$ and $q \in J^{+}(p)$ then the $q \in \mathcal{T}$.
\end{lemma}
\begin{proof}
    In spherical symmetry, Trapped surfaces ($\mathcal{T}$) are spacelike $S^2$ surfaces that hold conditions that the expansion scalar of outgoing and ingoing null geodesics ($\theta_{l}$ and $\theta_{n}$ respectively) must be negative.
    \begin{equation}
        \mathcal{T} = \{(t,r): \theta_{l} < 0 \,\,\,\ \& \,\,\,\ \theta_{n} < 0 \}
    \end{equation}
    As discussed in \textit{Lemma 1}, if $\dot{R}(t,r)$ or $\theta_{l}$ is negative for any null geodesics that geodesics will trapped into trapped surfaces. Therefore future of every non-spacelike geodesic of $\mathcal{A}$ is in trapped surfaces.
\end{proof}

In the case of null congruences, Caustic is the term commonly used to describe a singularity of congruence, not for the singularity in the spacetime structure \cite{wald}. As discussed in \cite{Poisson}, the caustics of the congruence of \textit{null generators} are the entry points into the event horizon. That caustic is not a caustic of null hypersurfaces and spacetime singularity. To broaden the concept, $\theta_{l}$ tends to positive and negative infinity, giving the caustic congruence of outgoing null geodesics. Congruences of null hypersurfaces can be defined in the same manner as congruences of null geodesics can. A congruence of null hypersurfaces in $\Phi$ (spacelike 3 surfaces) is a family of null hypersurfaces such that through each $S^2 \subset \Phi$ these passes precisely one null hypersurface in this family. However, if a trapped surface forms, then massive and mass-less particles will eventually end up at spacetime singularity and form a black hole.

\hspace{0.5cm}

\begin{theorem}
Consider a marginally bound LTB spacetime $(\mathcal{M},g)$, which is seeded by dust fluid and satisfies the energy conditions. Let $\mathcal{A}$ be considered an apparent horizon curve. Then singularity is to be hidden and future caustic of an outgoing null hypersurface will form iff the condition holds: $\frac{dt_{ah}}{dR}\Big|_{r = 0} < 0$.
\end{theorem}
\begin{proof} 
In that, $\exists$ a subset $\mathcal{X} \subset \mathcal{A}$, such that $q \in edge(\mathcal{X})$ and $q \in \Sigma$, the outgoing null hypersurface ($\Sigma$) passing through $q$, $\Sigma$ must be future inextendible (Here, inextendible hypersurface means every null geodesic which generate hypersurface are all inextendible.), and if the caustic of a null hypersurface form at $t_{s}(0)$. Then,
\begin{enumerate}
    \item A congruence of null hypersurfaces $\Sigma_{1}$ and $\Sigma_{2}$ intersect a subset $\mathcal{X}$, and
    \item The physical radius $R_{1}$ of $\Sigma_{1}$ and $R_{2}$ of $\Sigma_{2}$ at the $\theta_{l} \to -\infty$ (at the caustic of outgoing null hypersurfaces) always gives $(R_{1} - R_{2}) \to 0$ as $t \to t_{s}(0)$.
\end{enumerate}
Now, as we discussed in \textit{Lamma 2} at the future end of outgoing null geodesics $\frac{d t_{null}}{d R} = 0$. Now, we consider in theorem  $\frac{dt_{ah}}{dR}\Big|_{r = 0} < 0$, that is if $\frac{dt_{ah}}{dR}\Big|_{r = 0} < \frac{dt_{null}}{dR}\Big|_{r = 0} $. Therefore, $\exists$ a subset $\mathcal{X}$, with initial conditions on $\gamma$, $t_{0}<t_{s}(0)$, $\frac{dt_{null}}{dR}\Big|_{r = 0} =1$.

For a homogenous density profile at $r = 0$, the tangent field on the apparent horizon in the marginally bound LTB situation yields a value of $\frac{dt_{ah}}{dR}\Big|_{r = 0} = -2/3$. For more insight, we have analyzed null geodesics to show the caustic of spacetime singularity, as shown in Fig.~(\ref{ltbcollapsez1}). To study the light-like geodesic behavior near the singularity, we will study the null geodesic emanating from the center of a collapsing body ($r=0$) and time  $t_{o}$. For homogeneous case, $F(r) = M_{o}r^3$, integrating the eq.~(\ref{nullgeodesicsltb}) for the outgoing null geodesic, we get,
\begin{widetext}
\begin{equation}
    t_{null} = t_{o} + \frac{M_{o} r^3}{12} - \frac{1}{2} \sqrt{M_{o}} r^2 \left(1 - \frac{3\sqrt{M_{o}}}{2} t_{o}\right)^{1/3} + r \left(1 - \frac{3\sqrt{M_{o}}}{2} t_{o}\right)^{2/3}. \label{tnull}
\end{equation}
\end{widetext}
Here we have answered an important question that is, Is it possible for a null geodesic originating from the center ($r=0$) and, at time $t_{p}$, to fall into the singularity formation time, $t_{s}(R_{b})$? or Which null geodesic, reach the point $(t_{s}, R_{b})$? Here we have found such a $t_{o}$ that is $t_{p}$ value that the light emanating from that point will plunge into the singularity formation time, $t_{s}$. Therefore, we take condition $t_{null} = t_{s}$, and we get, 
\begin{equation}
    t_{p} = t_{s} - \frac{M_{o} R_{b}^3}{12},
\end{equation}
where, $t_{s} = 2/3\sqrt{M_{0}}$ and $R_{b}$ is the outer most co-moving shell. All the null geodesics that emanate from the center in between the $t_{p} < t < t_{s}$ will converge into a singularity formation time, $t_{s}$, as shown in Fig.~(\ref{ltbcollapsez1}) and Fig.~(\ref{penrose1}). The red solid in Fig.~(\ref{penrose1}) shows a null hypersurface that originates from the $t_{p}$, with the tangent of the null generator, $\frac{dt_{null}}{dR}\Big|_{r = 0} =1$. This is the unique behavior of null geodesics near the singularity in the black hole formation. For the homogeneous case, the equation of the event horizon can be easily calculated, by equating the boundary point of the apparent horizon from Eq.(\ref{ltbtimeah}) with Eq. (\ref{tnull}). We obtain $t_{o}$ for event horizon as follows:
\begin{equation}
   t_{eh} = t_{s} - \frac{27}{12} M_{o} R_{b}^{3}.
\end{equation}
Here $t_{eh}$ is the time at which the first light ray emitted from the central shell $r=0$ and its future endpoint is timelike infinity, $i^{+}$. In the inhomogeneous collapse, the singularity is hidden in the tangent field value between the interval $[-2/3, 1)$.
\end{proof}

\begin{theorem}
Consider a marginal bound LTB spacetime $(\mathcal{M},g)$, which is seeded by dust fluid and satisfies the energy conditions. Let $\mathcal{A}$ be considered an apparent horizon curve. Then, the caustic point of outgoing null hypersurfaces will never form iff $0 \leq \frac{dt_{ah}}{dR}\Big|_{r = 0} < \frac{dt_{null}}{dR}\Big|_{r = 0}$.
\end{theorem}
\begin{proof} 
As discussed in the proof of \textit{Theorem 1}, similarly from \textit{Lamma 2} at the future end of outgoing null geodesics $\frac{d t_{null}}{d R} = 0$. In the case of $0 \leq \frac{dt_{ah}}{dR}\Big|_{r = 0} < \frac{dt_{null}}{dR}\Big|_{r = 0} $, a subset $\mathcal{X}$ is empty set, $\mathcal{X} = \{\varnothing\}$. Therefore, the caustic of outgoing null hypersurfaces will never form. For zero caustic of outgoing null hypersurfaces, we take, $\frac{dt_{ah}}{dR} = 0$ in Eq.~(\ref{tangentah}) for inhomogeneous mass profile, for which
\begin{equation}
    M_{3} = 2 M_{0}^{5/2},
\end{equation}
as we have shown in the Fig.~(\ref{ltbcollapseInHomo4}).

The best example of a zero caustic point of the outgoing null hypersurface is given in \cite{Joshi:2023ugm}, where we have shown the formation of a future null singularity. Where, 
\begin{equation}
    t_{ah} = t_{s} = \frac{1}{3}\left((2M + R_{b})\sqrt{M+2R_{b}} -2M^{3/2}\right),\label{tahts}
\end{equation}
that is achieved by taking zero scale factor, $a(t=t_s)=0$. For the scale function and explicit expression of $F(t,r)$ given by in \cite{Joshi:2023ugm}. 
\begin{figure}[ht!]
\centering
\includegraphics[scale=.5]{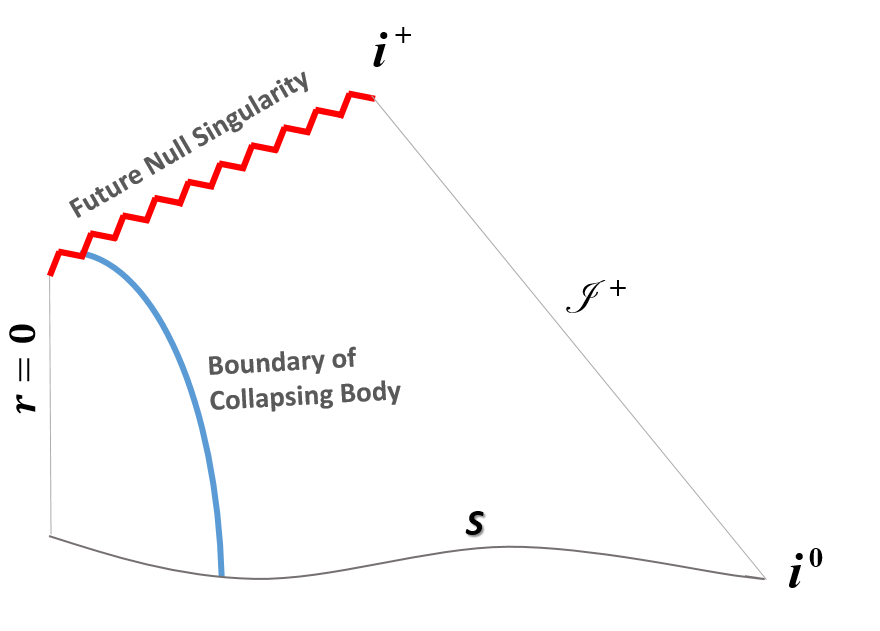}
\caption{Conformal diagram of Future-Null singularity spacetime}
\label{PenroseDiagram}
\end{figure}
As shown in Eq.~(\ref{tahts}) apparent horizon curve is independent of comoving radius $r$, therefore $\nexists$ $(t,r)\in (0,t_s)\times (0,R_b)$ satisfying  $F=R$. This suggests that the event and apparent horizon are absent. Fig.~(\ref{PenroseDiagram}) shows zero caustic point of the outgoing null hypersurface. 
\end{proof}

\begin{conj}
    At $r =0$, for generators of Cauchy horizon $\frac{dt_{null}}{dR}\Big|_{r \to 0} = 1$.
\end{conj}
In the inhomogeneous case, if the Cauchy horizon is present within the spacetime then we could argue that the naked singularity case arises, the numerical value of the mass function approaching the center is zero ($\lim_{r \to 0} F(r) = 0$). Now, from Eq.~(\ref{rdot}), the physical radius also goes to zero at the zero value of the comoving radius. However, the physical radius reaches zero at a particular time slice, $t_{s}$. Therefore, at $t_{s}$, $\dot{R}(t,r)$ is negatively infinite. Now, at $r =0$ and consider $t_{o}$ value is very close to the 
$t_{s}$ with condition $t_{o}<t_{s}$, that will provide $\dot{R}(t_{o},0) = 1$. From this condition, we conjectured that for the Cauchy horizon, tangent $\frac{dt_{null}}{dR}\Big|_{r = 0} =1$. Here, $t_{null}$ is co-moving time associated with $\gamma$, that is $\gamma(t_{null})$.
As suggested in Eq.~(\ref{CauchyD}), the future Cauchy horizon is the future boundary of Cauchy development. Therefore at $r \to 0$ and $R \to \epsilon$,
\begin{equation}
    t_{null} = t_{s}(0) + \epsilon
\end{equation}
Where $\epsilon$ is infinitesimally small value. In local visibility, the lower bound of the past incomplete null curve should not be smaller than the singularity formation time, $t_{s}(0)$, that is,
\begin{equation}
    \inf_{r\in \mathbb{R}^{+}} t_{null}(r) \geq t_{s}(0).\label{inf}
\end{equation}
By taking an initial value of outgoing null generators as $R = \epsilon$ and $t = t_{s}(0) + \epsilon$, and satisfy the inequality Eq.~(\ref{inf}) develop a Cauchy horizon. The red solid line is the Cauchy horizon, in Fig.~(\ref{ltbcollapseInHomo6}).




\begin{theorem}
Consider a marginally bound LTB spacetime $(\mathcal{M},g)$, which is seeded by dust fluid and satisfies the energy conditions. Let $\mathcal{A}$ be considered an apparent horizon curve. Then singularity is to be locally visible and Past-caustic of an outgoing null hypersurface form iff at least the condition holds: $\frac{dt_{ah}}{dR}\Big|_{r = 0} > 1$.
\end{theorem}
\begin{proof} 
As we have discussed in \textit{Conjecture 1}, the value of tangent on the null generator at the past end of outgoing null geodesics will be at least $\frac{d t_{null}}{d R} = 1$. We consider in theorem  $\frac{dt_{ah}}{dR}\Big|_{r = 0} > 1$, that is if $\frac{dt_{ah}}{dR}\Big|_{r = 0} > \frac{dt_{null}}{dR}\Big|_{r = 0} $. Therefore, $\exists$ a subset $\mathcal{X}$, with initial conditions $t_{0} = t_{s}(0)$ on $\gamma$, $\frac{dt_{null}}{dR}\Big|_{r = 0} \leq 1$.

Outer most boundary of $\mathcal{X}$ that is $\psi$, $\psi = \mathcal{CH}^{+} \cap \mathcal{A}$, is a two dimensional $\mathbb{S}^{2}$ surface embedded in the three dimensional spacelike submanifold of $\mathcal{M}$. Central singularity within LTB spacetime $\mathcal{S}_{LTB} = \{t= t_{s}(r), r \in [0, R_{b}]\}$. A point in spacetime diagram, $t_{s}(0) \in \mathcal{S}_{LTB}$, if it is locally visible (that is past caustic point of outgoing null hypersurface) in LTB metric which maps a line in a conformal diagram shown in a red sig-sag line, Fig.~(\ref{penrose4}), which describes the nature of past null singularity. Point on the Cauchy horizon that describe $S^{2}$ surface in the spherical symmetry, 
$$\Sigma_{t} \cap \mathcal{CH}^{+} \setminus \{p \in \mathcal{S}_{LTB}\} = p$$ $\forall p$ energy conditions are satisfied. Here $\Sigma_{t}$ is a spacelike hypersurface with constant co-moving time, $t$. So from that, we can conclude that on the Cauchy horizon, matter is well-behaved.
From the \textit{Theorem 3} we can say that in the inhomogeneous collapse with the mass function $F(r) = M_{0} r^{3} - M_{3}r^{6}$, for local visibility, the condition must satisfy at least the following inequality:
\begin{equation}
   \frac{dt_{ah}}{dR}\Big|_{r = 0} > 1, \,\,\,\,\,\,\,\,\,\,\ M_{3} > 5 M_{o}^{5/2}.\label{ineq}
\end{equation}
This suggests that the singularity will be globally or locally visible in the interval $(1, \infty)$.
\end{proof}
Null geodesic curve for the mass profile of $F(r) = M_{0} r^{3} - M_{3} r^{6}$ is given as follows:
\begin{widetext}
\begin{equation}
\begin{split}
    t_{null} = t_{o} + r \left(1 - \frac{3\sqrt{M_{o}}}{2} t_{o}\right)^{2/3} - \frac{1}{2} \sqrt{M_{o}} r^2 \left(1 - \frac{3\sqrt{M_{o}}}{2} t_{o}\right)^{1/3} + \frac{M_{o} r^3}{12} + \frac{M_{3} t_{o} r^{4}}{2 \sqrt{M_{o}} \left(1 - \frac{3\sqrt{M_{o}}}{2} t_{o}\right)^{1/3}} \\ \vspace{0.5cm}- \frac{12 M_{3} (-4 + 5 \sqrt{M_{0}} t_{o}) r^{5}}{120 \sqrt{M_{o}} \left(1 - \frac{3\sqrt{M_{o}}}{2} t_{o}\right)^{2/3}} - \frac{12 M_{3} (8 - 15 \sqrt{M_{o}} t_{o}) r^{6}}{720 (2 - 3 \sqrt{M_{o}} t_{o})} + \mathcal{O}(r^{7}). \label{tnull2}
    \end{split}
\end{equation}
\end{widetext}
Using Eq.~(\ref{tnull2}) we have computed null geodesics in Fig.~(\ref{InHomo}).
Regarding the formation of at least locally visible singularity, we have the following statement: Consider an unhindered gravitational collapse of a spherically symmetric perfect fluid. The singularity formed as an end state of such collapse is at least locally naked if and only if $\exists$ $X_0\in \mathbb{R}^{+}$ as a root of $V(X)$, where 
\begin{widetext}
    \begin{equation}\label{rootequation}
    V(X)=X-\frac{1}{\alpha}\left(X+\sqrt{\frac{F_{0}(0)}{X}}\left(\chi_1(0)+2r\chi_2(0)+3r^2\chi_3(0)\right)r^{\frac{5-3\alpha}{2}}\right)\left(1-\sqrt{\frac{F_0(0)}{X}}r^{\frac{3-\alpha}{2}}\right).
    \end{equation}
   Here
    \begin{equation}
        \alpha \in \left\{\frac{2n}{3}+1;\hspace{0.2cm} \hspace{0.2cm} n\in \mathbb{N} \right\}, \hspace{0.5cm}
    \chi_i(v)=\frac{1}{i!}\frac{\partial ^i t(r,v)}{\partial r^i}\bigg |_{r=0}, \hspace{0.5cm} F(r,v)=\sum_{i=0}^{\infty}F_i(v)r^{i+3},
    \end{equation}
\end{widetext}
and $t=t(r,v)$ is the time curve \cite{Joshi:1993zg}. To depict examples of past null singularities, we consider Misner-Sharp mass functions $F(r)$ given by
    \begin{equation}\label{exampleltbg}
        F(r)=F_0 r^3+ F_3 r^6; \hspace{0.3cm} F_0>0, \hspace{0.3cm} F_3<0. 
    \end{equation}
Corresponding $V(X)$ (Eq. (\ref{rootequation})) for such mass functions are obtained as: 
    \begin{equation}\label{rooteqnltb}
    \begin{split}
         V(X)=  2 X^2 + \sqrt{F_0}X^{3/2} - 3\sqrt{F_0} \chi_3(0) \sqrt{X} \\ + 3F_0 \chi_3(0),
    \end{split}
    \end{equation}
    where 
    \begin{equation}
        \chi_3(0)= -\frac{1}{2} \int_{0}^{1} \frac{F_3/a}{\left(F_0/a \right)^{3/2}}~ da,
    \end{equation}
Setting $F_0=1$, Eq. (\ref{rooteqnltb}) has positive real roots if and only if $F_3 < 25.967$, in which case, the singularity is visible \cite{Joshi:2023ugm}. This numerical value also satisfies the inequality given in Eq.~(\ref{ineq}), and the null geodesic curve (Eq.~(\ref{tnull2})) must satisfy Eq.~(\ref{inf}).

\section{genericity of zero caustic point}\label{sec3}
We defined the genericity of compact objects in terms of satisfying three main important conditions 1) Energy conditions in spacetime 2) Causal structure of spacetime in terms of global hyperbolicity 3) Strength of singularity in the sense of the Tipler strong singularity.

\subsection{Energy conditions}
The mass function considered ($F(r) = M_{0} r^{3} - M_{3} r^{6}$ and mass profile derived in \cite{Joshi:2023ugm}) is a monotonically increasing function and its satisfy all energy conditions. The result implies that free-falling observers co-moving with the dust particles measure the finite energy of the field, even arbitrarily close to the central singularity. In this sense, the solution of Einstein's field equation has a unique solution.

\subsection{Global hyperbolicity}
The cases in which the caustic of outgoing null hypersurfaces will
never form, they possess Cauchy surfaces. Cauchy horizon is absent in such cases. For example, as shown in Fig.~(\ref{ltbcollapseInHomo4}) and Fig.~(\ref{PenroseDiagram}). The spacelike hypersurface $\Sigma$ is present in the manifold that holds the condition of the domain of dependence, $\mathcal{D}(\Sigma) = \mathcal{D}^{+}(\Sigma) \cup \mathcal{D}^{-}(\Sigma)$, the $\mathcal{M}$ is globally hyperbolic.

\subsection{Strength of the null singularity}
In the gravitational collapse of the spherically symmetric spacetime described by eq.~(\ref{spacetime}), the following is the tangent field of the outgoing null geodesic as we approach singularity \cite{Mosani:2023vtr}:
\begin{equation}
    K^{t} = \frac{dt}{d\lambda} = \frac{\mathcal{P}}{R} = \frac{R}{\lambda}, \hspace{1cm} K^{r} = \frac{dr}{d\lambda} = \frac{\mathcal{P}}{R R^{'}} = \frac{r}{\lambda}.\label{ktkr}
\end{equation}
$\mathcal{P}$ is the functional form given in \cite{Mosani:2020ena}. For a singularity to be strong in the sense of Tipler \cite{Tipler:1977zza}, according to Clarke and Krolak \cite{Clarke}, the following inequality must hold along at least one null geodesic:
\begin{equation}
    \lim_{\lambda \to 0} \lambda^{2} R_{\mu \nu}K^{\mu}K^{\nu} > 0.
\end{equation}
For the metric given in equation (\ref{spacetime}), by using eq.~(\ref{ktkr}), we get,
\begin{equation}
    \left(\frac{2 R'(t,r)}{R((t,r)} \left(\dot{R}(t,r)\dot{R}^{'}(t,r) - R'(t,r) \ddot{R}(t,r)\right)\right)r^2 > 0. \label{inequal}
\end{equation}
By analyzing the inequality given in eq.~(\ref{inequal}), we can conclude that for homogeneous ($M_{o}r^{3}$) and inhomogeneous ($M_{o}r^{3} - M_{1}r^{6}$) cases in the marginally bound LTB collapse, the inequality condition holds, which directly implies that the singularity developed in both cases are Tipler strong singularity.  Similarly in the case of Future-null singularity, especially in the gravitational collapse of null singularity spacetime, the given condition in eq.~(\ref{inequal}), is hold.

\section{Results and discussion}\label{result}
The conclusions from this study can be summarised as follows:
\begin{itemize}
\item The null geodesic behavior was investigated near the time of singularity formation, $t_{s}(r)$. In the case of marginally bound LTB, the null geodesics that emanate from the center between the $t_{p}<t<t_{s}$ will enter the singularity formation time as seen in the fig.~(\ref{ltbcollapsez1}) which shows the Future caustic point of the outgoing null hypersurface. A caustic point where numerous null geodesics are emitted from one specific moment, which we defined as the Past-caustic point of the outgoing null hypersurface, will result if the inhomogeneity is sufficiently large, as can be seen in fig. (\ref{ltbcollapseInHomo6}). In the case of Future null singularity case, $t_{p} = t_{s}$. 

\item We analytically showed the local visibility of singularity in the marginally bound LTB metric and we have extensively discussed the nature of singularity using the tangent on the apparent horizon. For a homogenous density profile at $r = 0$, the tangent on the apparent horizon in the marginally bound LTB situation yields a value of $-2/3$. While in the case of inhomogeneity, the interval $(-2/3,1]$ gives the hidden singularity, and $(1,\infty)$ gives the local and global visible singularities. Since we were unable to demonstrate an analytical expression for the global visibility of the singularity in this instance, we came to the numerical conclusion that the globally visible singularity will be produced when $ \frac{dt_{ah}}{dR}\Big|_{r = 0} >30$. The caustic behavior is absent during the formation of null singularities, $t_{p} = t_{s}$.

\item We have analysed the zero caustic point in the outgoing null hypersurface which will give the possibility of the formation of future null and about the globally hyperbolic spacetime. The nature of the caustic point provides information about the central singularity. Additionally, we have analyzed the strength of the singularity as a result, we concluded that the singularity possesses zero caustic point could give Tipler's strong singularity as the final state.

\item The strongest alternative and solution to the causality criterion can be used to avoid the predictability issue by assuming that spacetime is globally hyperbolic. Since the region $J^{-}(\mathcal{J}^{+})$ in fig.~(\ref{PenroseDiagram}) is globally hyperbolic, weak cosmic censorship holds for the future null singularity by definition\cite{Held:1980gc}. Hence, the singularity can only be seen if the observer is at a timelike infinity, $i^{+}$. We discussed the causal structure of singularity spacetime for various cases using the expansion scalar of outgoing null geodesics and concluded that Future null singularity appears to be neither a black hole nor a naked singularity as shown in fig.~(\ref{PenroseDiagram}). Which possess zero caustic point in the outgoing null hypersurfaces.
\end{itemize}



\end{document}